\documentclass[a4paper,final,conference]{IEEEtran}
\usepackage{amsthm}
\usepackage{amsmath}
\usepackage{pifont}
\usepackage[top=0.58in, bottom=0.58in, left=0.6in, right=0.6in]{geometry}

\usepackage{enumerate}
\usepackage{amsfonts}
\usepackage{graphicx}
\usepackage{latexsym}
\usepackage{amssymb}
\usepackage{stmaryrd}
\usepackage{amscd}
\usepackage[small]{caption}
\usepackage{array}
\usepackage{tikz,ifthen,calc}
\usetikzlibrary{trees,snakes,automata}
\usetikzlibrary{arrows,shapes,positioning}
\usepackage{comment}
\usepackage{multicol}

\usepackage{color}

\newcommand{\cC}{{\cal C}}

\newcommand{\cH}{{\cal H}}

\newcommand{\sG}{\script{G}}

\newcommand{\sP}{\script{P}}

\newcommand{\bfa}{{\boldsymbol a}}
\newcommand{\bfb}{{\boldsymbol b}}
\newcommand{\bfc}{{\boldsymbol c}}

\newcommand{\bfe}{{\boldsymbol e}}

\newcommand{\bfh}{{\boldsymbol h}}

\newcommand{\bfs}{{\boldsymbol s}}

\newcommand{\bfu}{{\boldsymbol u}}
\newcommand{\bfv}{{\boldsymbol v}}

\newcommand{\bfx}{{\boldsymbol x}}
\newcommand{\bfy}{{\boldsymbol y}}
\newcommand{\bfz}{{\boldsymbol z}}

\newcommand{\bfX}{{\mathbf X}}

\newcommand{\sbinom}[2]{\left[ \begin{array}{c} #1 \\ #2 \end{array} \right] }

\DeclareMathAlphabet{\mathbfsl}{OT1}{cmr}{bx}{it}
\newcommand{\uuu}{\kern-1pt\mathbfsl{u}\kern-0.5pt}
\newcommand{\vvv}{\kern-1pt\mathbfsl{v}\kern-0.5pt}

\newcommand{\myboxplus}{\kern1pt\mbox{\small$\boxplus$}}

\makeatletter \DeclareRobustCommand{\sbinom}{\genfrac[]\z@{}}
\makeatother
\newcommand{\G}[2]{\sbinom{{#1}\kern-1pt}{{#2}\kern-1pt}}
\newcommand{\Gq}[2]{\sbinom{{#1}\kern-0.25pt}{{#2}\kern-0.25pt}}

\newcommand{\Ps}{\smash{{\sP\kern-2.0pt}_q\kern-0.5pt(n)}}
\newcommand{\sPs}{\smash{{\sP\kern-1.5pt}_q(n)}}
\newcommand{\Ptwo}{\smash{{\sP\kern-2.0pt}_2\kern-0.5pt(n)}}
\newcommand{\Ptwom}{\smash{{\sP\kern-2.0pt}_2\kern-0.5pt(m)}}
\newcommand{\Ptwonm}{\smash{{\sP\kern-2.0pt}_2\kern-0.5pt(n+m)}}
\newcommand{\Ptwoa}{\smash{{\sP\kern-2.0pt}_2\kern-0.5pt(1)}}
\newcommand{\Ptwob}{\smash{{\sP\kern-2.0pt}_2\kern-0.5pt(2)}}
\newcommand{\Ptwoc}{\smash{{\sP\kern-2.0pt}_2\kern-0.5pt(3)}}
\newcommand{\Ptwod}{\smash{{\sP\kern-2.0pt}_2\kern-0.5pt(4)}}
\newcommand{\Ptwoe}{\smash{{\sP\kern-2.0pt}_2\kern-0.5pt(5)}}
\newcommand{\Ptwof}{\smash{{\sP\kern-2.0pt}_2\kern-0.5pt(6)}}
\newcommand{\Ptwokm}{\smash{{\sP\kern-2.0pt}_2\kern-0.5pt(2k-1)}}
\newcommand{\Pone}{\smash{{\sP\kern-2.5pt}_2\kern-0.5pt(n{-}1)}}

\newcommand{\Gr}{\smash{{\sG\kern-1.5pt}_q\kern-0.5pt(n,k)}}
\newcommand{\Gi}{\smash{{\sG\kern-1.5pt}_q\kern-0.5pt(n,i)}}
\newcommand{\Gj}{\smash{{\sG\kern-1.5pt}_q\kern-0.5pt(n,j)}}
\newcommand{\Grmk}{\smash{{\sG\kern-1.5pt}_q\kern-0.5pt(n,n-k)}}
\newcommand{\Grdk}{\smash{{\sG\kern-1.5pt}_q\kern-0.5pt(2k,k)}}
\newcommand{\Grekappa}{\smash{{\sG\kern-1.5pt}_q\kern-0.5pt(n,e+1-\kappa)}}
\newcommand{\Grtwoekappa}{\smash{{\sG\kern-1.5pt}_q\kern-0.5pt(n,2e+1-\kappa)}}
\newcommand{\Gremkappa}{\smash{{\sG\kern-1.5pt}_q\kern-0.5pt(n,e-\kappa)}}
\newcommand{\Gn}{\smash{{\sG\kern-1.5pt}_2\kern-0.5pt(n,n{-}1)}}
\newcommand{\Gnq}{\smash{{\sG\kern-1.5pt}_q\kern-0.5pt(n,n{-}1)}}
\newcommand{\Gone}{\smash{{\sG\kern-1.5pt}_2\kern-0.5pt(n,1)}}
\newcommand{\Gqone}{\smash{{\sG\kern-1.5pt}_q\kern-0.5pt(n,1)}}
\newcommand{\GTwo}{\smash{{\sG\kern-1.5pt}_2\kern-0.5pt(n,k)}}
\newcommand{\GTwonk}[2]{{\smash{{\sG\kern-1.5pt}_2\kern-0.5pt({#1},{#2})}}}
\newcommand{\Gnk}{\smash{{\sG\kern-1.5pt}_2\kern-0.5pt(n,n{-}k)}}
\newcommand{\Greone}{\smash{{\sG\kern-1.5pt}_q\kern-0.5pt(n,e{+}1)}}
\newcommand{\Gretwo}{\smash{{\sG\kern-1.5pt}_q\kern-0.5pt(n,e{+}2)}}

\newcommand{\be}[1]{\begin{equation}\label{#1}}
\newcommand{\ee}{\end{equation}}

\newcommand{\Cref}[1]{Co\-rol\-la\-ry\,\ref{#1}}

\newtheorem{theorem}{Theorem}
\newtheorem{lemma}[theorem]{Lemma}

\newtheorem{proposition}[theorem]{Proposition}
\newtheorem{example}[theorem]{Example}

\newtheorem{problem}{Problem}
\newtheorem{problem*}{Problem}

\begin{document}

	\title{\textbf{On the Access Complexity of PIR Schemes}
		\hspace{-0.5ex}
		}
		\author{\IEEEauthorblockN{\textbf{Yiwei Zhang}}
		\IEEEauthorblockA{Dept. of Computer Science\\
			Technion\\
			Haifa 3200003, Israel \\
			\emph{ywzhang@cs.technion.ac.il}\vspace{-2.5ex}}
		\and
				\IEEEauthorblockN{\textbf{Eitan Yaakobi}}
		\IEEEauthorblockA{Dept. of Computer Science\\
			Technion\\
			Haifa 3200003, Israel \\
			 \emph{yaakobi@cs.technion.ac.il}\vspace{-2.5ex}}
		\and
						\IEEEauthorblockN{\textbf{Tuvi Etzion}}
		\IEEEauthorblockA{Dept. of Computer Science\\
			Technion\\
			Haifa 3200003, Israel \\
			 \emph{etzion@cs.technion.ac.il}\vspace{-2.5ex}}
		\and
	                    \IEEEauthorblockN{\textbf{Moshe Schwartz}}
		\IEEEauthorblockA{Department of ECE\\
			 Ben-Gurion Univ. of the Negev\\
			 Beer Sheva 8410501, Israel \\
			 \emph{schwartz@ee.bgu.ac.il}\vspace{-2.5ex}}}

\maketitle
\begin{abstract}
Private information retrieval has been reformulated in an information-theoretic perspective in recent years.
The two most important parameters considered for a PIR scheme in a distributed storage system are the storage
overhead and PIR rate. The complexity of the computations done by the servers for the various tasks of the distributed storage system
is an important parameter in such systems which didn't get enough attention in PIR schemes.
As a consequence, we take into consideration a third parameter, the \emph{access complexity} of a PIR scheme,
which characterizes the total amount of data to be accessed by the servers for responding to the queries
throughout a PIR scheme. We use a general covering codes approach as the main tool for improving
the access complexity. With a given amount of storage overhead, the ultimate objective is to characterize
the tradeoff between the rate and access complexity of a PIR scheme. This covering codes approach raises a new
interesting coding problem of generalized coverings similarly to the well-known generalized Hamming weights.
\end{abstract}

\section{Introduction}\label{sec:intro}
\renewcommand{\baselinestretch}{0.98}\normalsize\noindent

{\it Private information retrieval (PIR)} protocols, first introduced by Chor, Goldreich, Kushilevitz,
and Sudan in~\cite{CKGS98}, allow a user to retrieve a data item from a database without
revealing any information about the identity of the item to any single server. The original formulation of the PIR problem considers replicating a binary string on several non-communicating servers. The objective is to optimize the communication cost, including both the upload cost and the download cost, for privately retrieving one single bit. In recent years, the information-theoretic reformulation of the PIR problem assumes the more practical scenario in which the files are of arbitrarily large size. Under this setup, the number of uploaded bits can be neglected with respect to the corresponding number of downloaded bits since the upload does not depend on the size of the file~\cite{CHY15}. This reformulation introduces the \emph{rate} of a PIR scheme to be the ratio between the size of the retrieved file and the total number of downloaded bits from all servers. The supremum of achievable rates over all PIR schemes is defined as the \emph{PIR capacity}. In their pioneering work~\cite{SJ17B} Sun and Jafar determine the exact PIR capacity of the classical PIR model of replication.

Starting from~\cite{SRR14}, the research of PIR has been combined with distributed storage system instead of the replication-based system.
This brings in the other important parameter, i.e., the {\it storage overhead} of the distributed storage system, defined as the ratio between the total number of bits stored on all the servers and the number of bits of the database. Several papers have been studying the relation between the storage overhead and the rate of a PIR scheme.
Chan et al.~\cite{CHY15} offer a tradeoff between the storage overhead and rate for linear PIR schemes. They show that when each server stores a fraction $0<\epsilon\le 1$ of the database, then the rate of a linear PIR scheme should be at most $\frac{N-1/\epsilon}{N}$, where $N$ is the number of server. Tajeddine et al.~\cite{TGE17} propose a PIR scheme achieving this upper bound when the storage code is an arbitrary $(N,K)$-MDS code, so $\epsilon=\frac{1}{K}$ and the PIR rate is $\frac{N-K}{N}$. Banawan and Ulukus~\cite{BU16} show that the exact PIR capacity when using an arbitrary $(N,K)$-MDS storage code is $(1+\frac{K}{N}+\cdots+\frac{K^{M-1}}{N^{M-1}})^{-1}$, a value dependent on the number of files $M$ and tends to $\frac{N-K}{N}$ when $M$ approaches infinity. However, similar to the scheme of Sun and Jafar~\cite{SJ17B}, this optimal scheme can be implemented only if the file size $L$ is an exponential function of $M$~\cite{SJ16C,XZ17}. For a more practical setting we are more interested in the case when  $L$ is at most a polynomial value in terms of $M$ and the scheme of Tajeddine et al.~\cite{TGE17} works for this setup.

Recall the development of the research on distributed storage systems:
Besides optimizing repair bandwidth or storage for distributed storage
systems, \emph{access complexity} is also a concern since the time of
reading data may cause a bottleneck.  The research of optimal-access
MDS codes started in \cite{TWB14}. A similar idea in locally
repairable codes was introduced in~\cite{GHSY12} for the sake of reducing
the nodes to be accessed.  The complexity of the computations done by
the servers for the various tasks of the distributed storage system is
an important parameter in such systems which didn't get enough
attention in PIR schemes.  The only work which took the computational
complexity of the servers, in the new PIR model, into account, was
done by Lavauzelle~\cite{Lav18}. Our approach is completely different.
For practical use of PIR protocols in distributed storage systems, we
should also consider the access complexity in the scheme. However, to
the best of our knowledge, the access complexity of PIR has not been
studied in previous works so far.  In fact, most known PIR schemes
require accessing almost all of the data stored on each server in the
worst case.  The next example demonstrates the concepts and
improvements for the access complexity that we study in this work.  We
will consider the worst case in this paper, but the average case is
also very interesting from a theoretical and practical points of view.

\begin{example}
{\it Consider the following 2-server PIR scheme where each server stores the whole database $\bfx=(\bfx^1,\bfx^2,\dots,\bfx^M)$. A user chooses an arbitrary binary vector $\bfa=(a_1,\dots,a_M)\in \mathbb{F}_2^{M}$ and then sends $\bfa$ and $\bfa+\bfe_f$ to the two servers respectively. From the responses $\sum_{i=1}^Ma_i\bfx^i$ and $\sum_{i=1}^Ma_i\bfx^i+\bfx^f$ the user successfully retrieves the desired file $\bfx^f$ privately. While the main advantage of this solution is its low download complexity, it suffers from extremely large access complexity since in the worst case almost all $M$ files are accessed on each server. Hence, in this scheme the bottleneck will no longer be the upload or download time, but the access time to read all files. The access complexity can be improved at the cost of increasing the storage overhead. That is, when storing more information in the servers the computation $\bfa \cdot \bfx$ will require to access a fewer number of files. For example, assume we also store in each server the file $\bfx_\Sigma$ given by $\bfx_\Sigma=\sum_{i=1}^M\bfx^i$. Then, in the worst case, the server will read only $M/2$ files. Thus we save half of the access complexity in the tradeoff of storing one additional file on each server.}
\end{example}

Intuitively for a PIR scheme in a distributed storage system there will be a relationship among the three parameters: storage overhead, PIR rate, and access complexity. The ultimate objective is to characterize the tradeoff of any two parameters when fixing the third. In this paper, we make a first step towards solving this problem. Given the number of servers $N$, the number of files $M$, we fix the size of the storage space for each server (and thus fix the storage overhead) and analyze the rate and access complexity of several PIR schemes.

The rest of the paper is organized as follows. In Section~\ref{sec:stat}, we give a formal statement of the PIR problem studied in the paper. In Section~\ref{sec:cov}, we discuss how to improve the access complexity using covering codes. In Section~\ref{sec:main}, we analyze the rate and access complexity for several PIR schemes. Finally, Section~\ref{sec:concl} concludes the paper.

\section{Problem Statement}\label{sec:stat}

A PIR scheme for a distributed storage system consists of the following parameters:

\begin{itemize}
  \item The system has $N$ servers. A database consists of $M$ files $\bfx^1,\bfx^2,\dots,\bfx^M$ of equal length $L$. The size of the database is then $ML$.

  \item Each server stores $\epsilon ML$ bits, $\epsilon>0$. Thus the total storage is $\epsilon NML$.

  \item The {\it storage overhead} is defined as the ratio between the total storage and the size of the database, i.e., $\epsilon N$.

  \item The storage code of the system is an encoding mapping $(\bfx^1,\bfx^2,\dots,\bfx^M)\in\mathbb{F}_2^{ML}\longrightarrow(\bfy_1,\dots,\bfy_N),~\bfy_n\in\mathbb{F}_2^{\epsilon ML}.$

  \item To retrieve a file, a user downloads $\rho_n$ bits from the $n$th server. The total download cost is then $\sum_{n=1}^N \rho_n$.

  \item The {\it rate} $\Omega$ of a PIR scheme is defined as the ratio of the size of a desired file and the number of downloaded bits, i.e. $\Omega=\frac{L}{\sum_{n=1}^N \rho_n}$.

  \item The $\rho_n$ downloaded bits from the $n$th server are functions of the data $\mathbf{y}_n$ it stores. The calculation of these downloaded bits requires the server to access $\delta_n ML$ bits in $\mathbf{y}_n$. $\delta_n$ is called the \emph{access complexity} of the $n$th server. The {\it total access complexity} is defined as $\Delta=\sum_{n=1}^N \delta_n$.
\end{itemize}

We call a 6-tuple $(N,M,L,\Omega,\Delta,\epsilon)$ {\it achievable}, if for a distributed storage system with parameters $N$, $M$ and $L$, we have a PIR scheme with rate $\Omega$, total access complexity $\Delta$, and each server stores a fraction $\epsilon>0$ of the whole database (and thus the storage overhead is $\epsilon N$). When $N$, $M$ and $L$ are clear from the context or not relevant, we abbreviate the 6-tuple as a 3-tuple $(\Omega,\Delta,\epsilon)$. The ultimate objective is to characterize the exact tradeoff between any two of the parameters $\Omega$, $\Delta$ and $\epsilon$ when fixing the third. In this paper we make a first step towards solving this problem by finding some achievable 3-tuples of $(\Omega,\Delta,\epsilon)$ with a predetermined $\epsilon$.

Intuitively the storage space can be divided into two parts. One part represents the indispensable storage for a particular PIR scheme and is referred to as {\it the storage for PIR}. This part represents the independent symbols stored on each server. The remaining part is jointly designed with the former part for improving the access complexity on each server. We illustrate this idea via the following example.

\begin{example}
{\it Consider a distributed storage system storing a database containing $M$ files $\bfx^1,\bfx^2,\dots,\bfx^M$ of equal size $L$. Assume we have $N=3$ servers with $\epsilon=1$, i.e., each server can store $ML$ bits. A user wants to retrieve a specific file $\bfx^f$.

One way is to allocate all the storage space to be used for PIR, so each server stores the whole database. Divide each file into two equal parts $\bfx^m=(\bfx^m_1,\bfx^m_2)$. A user chooses two independent random vectors $\bfa$ and $\bfb$ in $\mathbb{F}_2^{M}$. He asks for $\sum_{i=1}^Ma_i\bfx^i_1+\sum_{i=1}^Mb_i\bfx^i_2$, $\sum_{i=1}^Ma_i\bfx^i_1+\sum_{i=1}^Mb_i\bfx^i_2+\bfx^f_1$ and $\sum_{i=1}^Ma_i\bfx^i_1+\sum_{i=1}^Mb_i\bfx^i_2+\bfx^f_2$ from the three servers respectively. Therefore he downloads $\frac{3L}{2}$ bits, so the rate of the scheme will be $\Omega=2/3$. Each server will access almost all the data in the worst case. Altogether almost $3ML$ bits should be accessed throughout the scheme. Then the total access complexity will be $\Delta=3$. So we have an achievable 3-tuple $(\Omega=2/3,\Delta=3,\epsilon=1)$.

Yet another way is to only use half of the storage for PIR and the other half for improving the access complexity. Let each server store only half of the database. Say we have $\{\bfx^m_1:1\le m \le M\}$ on the first server, $\{\bfx^m_2:1\le m \le M\}$ on the second server and a coded form $\{\bfx^m_1+\bfx^m_2:1\le m \le M\}$ on the third server. Again a user chooses two independent random vectors $\bfa$ and $\bfb$ in $\mathbb{F}_2^{M}$. He makes two queries from each server and gets the responses as follows:
\begin{center}
\begin{small}
$\begin{array}{ccc}
  \text{Server I} & \text{Server II} & \text{Server III}  \\\hline
  \sum_{i=1}^Ma_i\bfx^i_1+\bfx^f_1 & \sum_{i=1}^Ma_i\bfx^i_2 & \sum_{i=1}^Ma_i(\bfx^i_1+\bfx^i_2) \\
  \sum_{i=1}^Mb_i\bfx^i_1 & \sum_{i=1}^Mb_i\bfx^i_2+\bfx^f_2 & \sum_{i=1}^Mb_i(\bfx^i_1+\bfx^i_2) \\\hline
\end{array}$
\end{small}
\end{center}
This is exactly the scheme of Tajeddine et al. in~\cite{TGE17} when using a $(3,2)$-MDS storage code. In this scheme the download will be $3L$ bits so the rate will be $\Omega=1/3$. To improve the access complexity, each server stores a coded form of the data using a covering code approach instead of storing $\{\bfx^m_1:1\le m \le M\}$, $\{\bfx^m_2:1\le m \le M\}$ or $\{\bfx^m_1+\bfx^m_2:1\le m \le M\}$ in their original form. For each query a server will only need to read about $0.22ML$ bits (to be explained in Section \ref{sec:cov}). So altogether at most $1.32ML$ bits are accessed in the scheme, resulting in the total access complexity $\Delta=1.32$. So we have an achievable 3-tuple $(\Omega=1/3,\Delta=1.32,\epsilon=1)$.}
\end{example}

\section{Access Complexity Using Covering Codes}
\label{sec:cov}

A (binary) {\it covering code} $\cC$ of length $\ell$ with {\it covering radius} $R$ is a set of vectors in $\{0,1\}^\ell$
such that for every vector $\bfu\in\{0,1\}^\ell$ there exists a codeword $\bfc\in\cC$
with Hamming distance $d_H(\bfc,\bfu)\leq R$. Covering codes were extensively studied and comprehensive information on them can be found in~\cite{ccbook}. For linear covering codes this property can be translated as follows.

\begin{proposition}{\rm\cite{ckms85}}
Let $\cC$ be a linear code of length $\ell$, dimension $k$, redundancy $r=\ell-k$, and a parity check matrix~$\cH$ of
size $r\times \ell$. Then, $\cC$ is a covering code with covering radius~$R$ if and only if for every column vector $\bfs\in\{0,1\}^r$
there exists a row vector $\bfy \in \{0,1\}^{\ell}$ of Hamming weight at most $R$, such that $\cH \cdot \bfy^T = \bfs$.
\end{proposition}

The other way to explain the covering radius of a linear code is as follows.
A column vector $\bfs\in\{0,1\}^r$ is actually a syndrome corresponding
to a particular coset of the code~$\cC$ in~$\mathbb{F}_2^{\ell}$. In this coset
one can find a vector $\bfy\in\mathbb{F}_2^{\ell}$
(not necessarily unique) with minimum Hamming weight. The vector $\bfy$ is
known as a {\it coset leader} and its weight is known
as the {\it coset weight}. Then one can get the vector $\bfs$ by summing up
the columns of $\cH$ indexed by the support set of $\bfy$.
Thus the covering radius of a linear code is exactly the maximum of all its
coset weights. Linear covering codes can be used to improve the access complexity as follows.

Suppose we have a database $\bfx$ which can viewed as a $t \times r$ matrix, i.e.
$\bfx=(\bfx_1,\dots,\bfx_r)$, where each~$\bfx_{i}$, ${1 \le i \le r}$,
is a column vector of length $t$.
Let $\cC$ be a linear code of length $\ell$, dimension $k$, redundancy $r=\ell-k$,
covering radius~$R$ and an $r \times \ell$ parity check matrix $\cH=[\bfh_1,\dots,\bfh_{\ell}]$.
Each server stores the database $\bfx$ encoded by the columns of $\cH$.
That is, the server stores $\ell$ column vectors $\bfz_i=\bfx \cdot \bfh_i$
for $1\leq i\leq \ell$. In other words, $\bfz_i$ is a linear combination of the files (column vectors) of the database $\bfx$,
with coefficients taken from $\bfh_i$. The user who chooses
an arbitrary binary column vector $\bfs=(\bfs_1,\ldots,\bfs_r)^T$, for the query, to
the $j$-th server, wants to retrieve from the $j$-th server the vector $\bfx \cdot \bfs$.
To compute $\bfx \cdot \bfs$, the server first finds the coset
leader~$\bfy$ such that ${\cH \cdot \bfy^T=\bfs}$. Then computing $\bfx \cdot \bfs$ is equivalent to
$$
\bfx \cdot \bfs = \bfx \cdot (\cH \cdot \bfy^T) =\bfx \cdot \left( \sum_{i:y_i=1}  \bfh_i \right)= \sum_{i:y_i=1} \bfz_i.
$$
Since the Hamming weight of the coset leader $\bfy$ is at most~$R$, it follows that
we only need to access at most $R$ columns of $\cH$ to compute $\bfx \cdot \bfs$. Moreover, $\cH$ can be chosen in the form $\cH=[I_r~|~A_{r\times(\ell-r)}]$ and thus we can always have a systematic form of the original data.
The asymptotic connection between the length $\ell$, covering radius $R$, and the
dimension $k$ of the linear covering code can be roughly estimated by the sphere-covering bound
$$2^{k}\cdot 2^{H(R/\ell)\ell} \approx 2^\ell,$$
or
$$\frac{k}{\ell}+H(R/\ell) = 1,$$
so $H(R/\ell) = 1-\frac{k}{\ell} = \frac{r}{\ell}$, where $H(\cdot)$ is the binary entropy function.
By setting the covering radius to be $R=\alpha r$ and the size of the storage $\ell=\beta r$, we have
\begin{equation} \label{coveringcode}
H\left(\frac{\alpha}{\beta}\right) = \frac{1}{\beta}.
\end{equation}
Solving this equation and the relation between $\alpha$ and $\beta$ can be represented as a function $\alpha=f(\beta)$, depicted in Fig.~\ref{fig:PIR}.

\medskip

\begin{figure}[h]
\centering
\vspace{-3ex}
\includegraphics[width=0.9\linewidth]{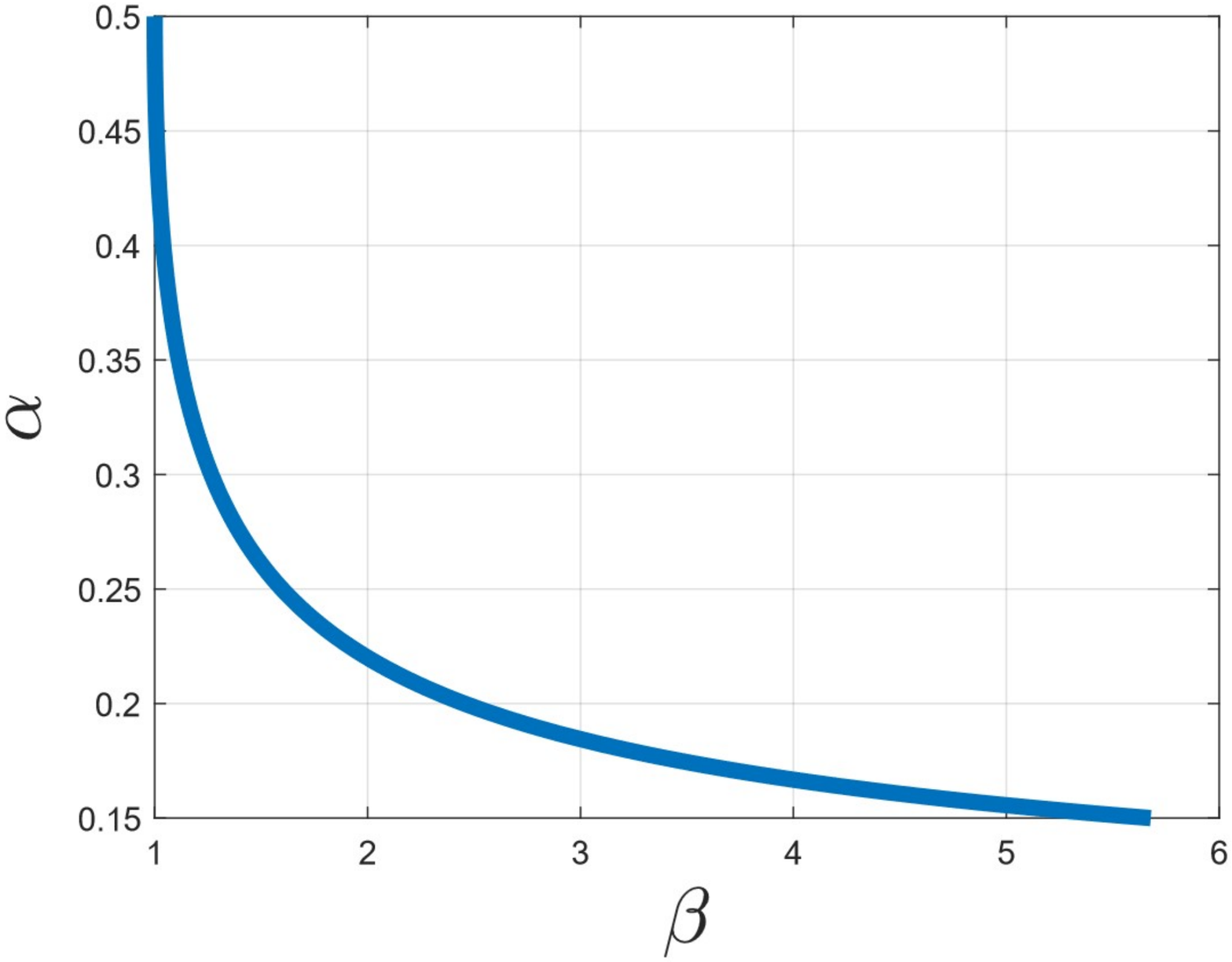}%
\caption[Access vs. Storage]
{Access vs. Storage}\vspace{-3ex}
\label{fig:PIR}%
\end{figure}

\medskip

It is not an easy task to explicitly construct linear codes achieving the above sphere-covering bound.
However, the existence of such codes has been proved:

\begin{proposition}{\rm\cite{CF85}}
  Let $0\le R\le \ell/2$. Then there exists an infinite sequence of linear codes $C_\ell$ of growing length $\ell$ with covering radius $R(C_\ell)\rightarrow R$ and the code rate is between $1-H(R/\ell)$ and $1-H(R/\ell)+O(\ell^{-1}\log_2\ell)$.
\end{proposition}

In a PIR scheme which does not take access complexity into consideration, the data
stored on a server can be usually represented as $r$ independent strings
of the same length and a query asks for a linear combination of these strings.
Using the covering code approach above, we may store a coded version of the data
instead of only storing the original form. The access complexity could be improved as follows.

\begin{theorem}
\label{thm:cov}
Suppose there exists a linear binary covering code with redundancy $r$, code length $\beta r$ and covering radius $\alpha r$. Given a set of $r$ independent strings $\{\bfx_1,\dots,\bfx_r\}$, a server can store a coded form of these strings as $\{\bfz_1,\dots,\bfz_{\beta r}\}$, such that computing any linear combination of $\{\bfx_1,\dots,\bfx_r\}$ requires only accessing at most $\alpha r$ substrings in $\{\bfz_1,\dots,\bfz_{\beta r}\}$. The asymptotic relation of $\alpha$ and $\beta$ is $H(\frac{\alpha}{\beta})=\frac{1}{\beta}$.
\end{theorem}

A further remark is that adding redundancies in the storage does not affect
the privacy of the original scheme. In essence the privacy is only related to the set of queries.
Finally, we note that the problem of reducing the access complexity when
replying to queries of the form mentioned in this section is not relevant only for PIR schemes.
The approach can be relevant for other models which require this or similar computation.
An example for such a problem was studied in~\cite{HBA98} for the partial-sum problem where the authors also used covering codes. However, since the computations involved integer numbers, the storage overhead was exponential with the number of items.

\section{PIR Rate vs. Access Complexity}\label{sec:main}

Now we begin to analyze the PIR rate and access complexity for two kinds of PIR schemes,
a scheme by Tajeddine et al.~\cite{TGE17} and a scheme
of Blackburn, Etzion and Paterson ({B-E-P} scheme)~\cite{BEP17}. Given $\epsilon$ indicating the
size of the storage space of each server, we first choose some $0 \le \pi \le \epsilon$ indicating
that the size of the storage for PIR, i.e., the amount of storage of independent symbols.
Using this $\pi$ fraction of storage we implement a proper PIR scheme with high rate.
Then we analyze the total access complexity of this scheme by making use of the remaining $\epsilon-\pi$ fraction of the storage space.

\subsection{The scheme of Tajeddine et al.~\cite{TGE17}}

When $\pi=\frac{1}{K}$, $K<N$, the rate of the scheme of Tajeddine et al.~\cite{TGE17} achieves the upper bound $\Omega=\frac{N-K}{N}$ proposed by Chan et al. in~\cite{CHY15}. So we begin with analyzing how to improve the access complexity of this scheme using the covering code approach.

Recall the framework of the scheme. Let each file ${\bfx^m\in\mathbb{F}_2^L}$ be represented
in the form of a matrix $\bfX^m=\{\bfx^m_{i,j}:{1\le i \le N-K}, 1\le j \le K\}$, where
each $\bfx^m_{i,j}$ represents a binary substring of length $\frac{L}{K(N-K)}$. Let $\Lambda=[\lambda_1,\dots,\lambda _N]$ be a $K\times N$ generator matrix of the storage code. Then the $n$th server stores $\bfX^m\lambda_n$, which
are $N-K$ linearly independent substrings as functions of $\bfx^m$. The whole storage
on the $n$th server is thus a concatenation of altogether $M(N-K)$ linearly independent substrings.
Each server will receive $K$ queries, where each query asks for a certain linear combination of these $M(N-K)$ substrings.

We make use of the additional storage of size $(\epsilon-\frac{1}{K})ML$ bits on each server.
Select a covering code with redundancy $r=M(N-K)$ with code length $\beta r$ where $\beta=K\epsilon$.
Instead of storing the $r$ substrings of $\bfy_n$ in their original form,
the server stores $\beta r$ substrings according to the covering code approach.
Then by Theorem \ref{thm:cov}, to answer each query, the server only needs
to access at most $\alpha r=f(\beta)r$ substrings. Recall that each server
receives $K$ queries. Thus each server will access at most $\min\{f(\beta)rK,r\}$ substrings,
since accessing the $r$ linearly independent substrings are already enough for computing any linear combination.
Therefore, the total number of bits accessed by each server is $\min\{ f(K\epsilon)ML,\frac{ML}{K}\}$
and thus the total access complexity over all servers will be $\Delta=\min\{ Nf(K\epsilon), \frac{N}{K}\}$.

For example, select $N=10$. Let $\epsilon=1$. We can take arbitrary $1\le K \le 9$ and apply the scheme of Tajeddine et al. The PIR rate and total access complexity are listed as follows. $\Delta'=\frac{N}{K}$ corresponds to the total access complexity when there is no redundancy in each server.

\begin{center}
$\begin{array}{cccc}
  K & \Omega=\frac{N-K}{N}  & \Delta & \Delta'\\\hline
  1  & 0.9  &  5.000 & 10.000\\\hline
  2  & 0.8  &  2.201 & 5.000\\\hline
  3  & 0.7  &  1.845 & 3.333\\\hline
  4  & 0.6  &  1.668 & 2.500\\\hline
  5  & 0.5  &  1.556 & 2.000\\\hline
  6  & 0.4  &  1.477 & 1.667\\\hline
  7  & 0.3  &  1.418 & 1.429\\\hline
  8  & 0.2  &  1.250 & 1.250\\\hline
  9  & 0.1  &  1.111 & 1.111\\\hline
\end{array}$
\end{center}

The table above indicates the covering code approach does improve the total access complexity for $K\neq 8,9$.
We further note that Tajeddine et al. mention that their scheme could be implemented
with cutting each file into $\frac{\text{l.c.m.}(K,N-K)}{K}\times K$ substrings instead
of $(N-K)\times K$, and correspondingly the number of subqueries for each server
is $\frac{\text{l.c.m}(K,N-K)}{N-K}$ instead of $K$. This modification allows us to further
improve the access complexity. Select a covering code with redundancy $r'=M\frac{\text{l.c.m.}(K,N-K)}{K}$ with
code length $\beta r'$ where $\beta=K\epsilon$. Then by Theorem \ref{thm:cov} each
query will access at most $\alpha r'=f(\beta)r'$ substrings. Each server
receives $\frac{\text{l.c.m}(K,N-K)}{N-K}$ queries and thus the number of bits to be
accessed on each server is at most $f(\beta)r'\frac{\text{l.c.m}(K,N-K)}{N-K}\times\frac{L}{K\times\frac{\text{l.c.m.}(K,N-K)}{K}}
=f(K\epsilon)ML\frac{\text{l.c.m.}(K,N-K)}{K(N-K)}.$
Therefore when $K$ and $N-K$ are not coprime, the total access complexity will be further improved as $\Delta=Nf(K\epsilon)\frac{\text{l.c.m.}(K,N-K)}{K(N-K)}=\frac{Nf(K\epsilon)}{\gcd(K,N-K)}$. Thus some of the results in the example above can be improved as follows.

\begin{center}
$\begin{array}{ccc}
  K & \Omega=\frac{N-K}{N}  & \Delta  \\\hline
  2  & 0.8  &  1.100 \\\hline
  4  & 0.6  &  0.834 \\\hline
  5  & 0.5  &  0.311 \\\hline
  6  & 0.4  &  0.739 \\\hline
  8  & 0.2  &  0.685 \\\hline
\end{array}$
\end{center}

In conclusion, applying the scheme of Tajeddine et al. results in several achievable 3-tuples as follows.
\begin{theorem}
In a distributed storage system consisting of $N$ servers, for every $1\le K <N$, $\epsilon\ge \frac{1}{K}$, the tuple $(\Omega=\frac{N-K}{N},\Delta=\min\{\frac{Nf(K\epsilon)}{\gcd(K,N-K)},N/K\},\epsilon)$ is achievable.
\end{theorem}

We close this subsection by discussing the possibility of further improving the access complexity for the scheme of Tajeddine et al.
Note that each server may receive multiple queries. The data accessed by a server when
responding to different queries may have certain overlap. Reconsider Example 1, if we
have two queries, then the server may read only $\frac{3M}{4}$ files instead of
reading $\frac{M}{2}$ files twice. Further improving the access complexity for the scheme
of Tajeddine et al. (by taking advantage of possible overlap when reading multiple queries)
relies on a good solution to the coding theoretic problem presented in the next subsection.

\subsection{Generalized coset weights}

{Given a binary linear code, for every $\tau$ cosets of the code, choose one vector from each coset and find the size of the union of their support sets.
The minimum of this value is called the {\it $\tau$-coset weight} of these $\tau$ cosets. What is the maximum value of all $\tau$-coset weights? When $\tau=1$ this is the covering radius $R$ of the linear code. When $\tau \ge 2$, we would like to see weights smaller than $\tau R$.}

The $\tau$-coset weights (for covering) are akin to the generalized Hamming weights
(for distance) defined in~\cite{Wei91} which
were considered in hundreds of papers.

Let $[n,k,d]$ code denote a binary linear code of length $n$, dimension $k$, and minimum Hamming distance $d$.

\begin{lemma}
The $\tau$-coset weight of a code $\cC$ is the minimum number of columns $\ell$ in the parity check matrix $\cH$ of $\cC$,
such that for each $\tau$ syndromes of $\cC$, there exists a set of $\ell$ columns of $\cH$ which has
$\tau$ linear combinations of this set to form these $\tau$ syndromes.
\end{lemma}

\begin{theorem}
The $\tau$-coset weight of an $[n,k,d]$ code $\cC$ is at most $n-k$ for each $\tau \geq 1$.
\end{theorem}
\begin{proof}
The parity check matrix $\cH$ of $\cC$ has $n-k$ linearly independent columns. A set of such $n-k$ columns
covers a word in each coset of $\cC$.
\end{proof}

\begin{theorem}
The $\tau$-coset weight of the $[2^m-1,2^m -1 -m,3]$ Hamming weight is $\tau$ for each $1 \leq \tau \leq m$.
\end{theorem}

\begin{theorem}
The $\tau$-coset weight of the $[2^m,2^m -1 -m,4]$ extended Hamming weight is $\tau +1$ for each $1 \leq \tau \leq m$.
\end{theorem}

For many types of BCH codes with minimum distance $d$ and covering radius $R$ we have proved that
the 2-coset weight is smaller than $2R$. This was generalized in some cases for $\tau$-coset weights
with $\tau >2$. This and other related results will be considered in the full version of this paper.

\subsection{Using several parallel B-E-P schemes}

Consider the scheme of Tajeddine et al. when $K=1$, i.e., replicated databases. Each file $\bfx^m\in\mathbb{F}_2^L$ is divided into $N-1$ substrings $\bfx^m_1,\dots,\bfx^m_{N-1}$ of length $\frac{L}{N-1}$. Each server stores $\bfy=\{\bfx^1_1,\dots,\bfx^1_{N-1},\bfx^2_1,\dots,\bfx^2_{N-1},\dots,\bfx^M_1,\dots,\bfx^M_{N-1}\}$, altogether $(N-1)M$ substrings. A user chooses a random binary vector $\bfv$ of length $(N-1)M$. The $N$th server receives the query vector $\bfv$ and the $n$th server receives the query vector $\bfv+\bfe_{(f-1)(N-1)+n}$, $1\le n \le N-1$, where $f$ is the index of the desired file. Then from the response of the $n$th server and the $N$th server, the user retrieves the string $\bfx^f_n$, $1\le n \le N-1$.

The B-E-P scheme recently proposed by Blackburn, Etzion and Paterson suggests a different way, whose original motivation is to optimize the upload complexity of the query vectors. A user who wants to retrieve the file $\bfx^f$ chooses $M$ elements $z_1,\dots,z_M\in\mathbb{Z}_{N}$ uniformly and independently at random. The $n$th server receives $(b_{1n},\dots,b_{Mn})$ where $b_{fn}=z_f+n\pmod{N}$ and $b_{mn}=z_m$ for $m\neq f$ and then responds with $\sum_{m=1}^{M} \bfx^m_{b_{mn}}$, where $\bfx^m_{0}$ represents the all-zero vector.

The main difference is that a query in the former scheme asks for an arbitrary linear combination of all the $(N-1)M$ substrings while a query in the latter scheme asks for a linear combination with a restricted pattern, i.e., at most one substring from each file is involved in the linear combination. This restriction may allow for a better way to improve the access complexity than the covering code approach.

\begin{example}
$N=3$, $M=3$. Consider the necessary amount of storage overhead for a PIR scheme with rate $2/3$ and total access complexity $1$.
The scheme of Tajeddine et al. gives an achievable tuple $(2/3,1,13/6)$, which
applies a covering code of length 13, redundancy 6 and covering radius 2 to store the six substrings
$\{\bfx^1_1,\bfx^1_2,\bfx^2_1,\bfx^2_2,\bfx^3_1,\bfx^3_2\}$. $\epsilon=13/6$ cannot
be improved for the scheme of Tajeddine et al. since 13 is the minimum length of a linear covering
code with redundancy 6 and covering radius 2~\cite[p. 202]{ccbook}.
However, if we use the B-E-P scheme, then each server can store the following 11 substrings:
$\{\bfx^1_1,\bfx^1_2,\bfx^2_1,\bfx^2_2,\bfx^3_1,\bfx^3_2,\bfx^1_1+\bfx^2_2,
\bfx^2_1+\bfx^3_2,\bfx^3_1+\bfx^1_2,\bfx^1_1+\bfx^2_1+\bfx^3_1,\bfx^1_2+\bfx^2_2+\bfx^3_2\}$.
This already guarantees that we only need to read at most two substrings for any
query in the B-E-P scheme. Thus the B-E-P scheme gives an achievable tuple (2/3,1,11/6).
\end{example}

In the former two subsections, on each server $\frac{1}{K}ML$ bits are allocated as the storage for PIR and the remaining ${(\epsilon-\frac{1}{K})ML}$ bits are designed for improving the access complexity. Now consider the case when $\frac{p}{q}ML$ bits are allocated for the PIR scheme and the remaining $(\epsilon-\frac{p}{q})ML$ bits are used for improving the access complexity, where $\frac{1}{N}\le\frac{p}{q}\le \epsilon$ cannot be simplified to the form $\frac{1}{K}$. Once we have a proper PIR scheme with good rate in this setup, the idea for improving the access complexity will be exactly the same approach aforementioned.

As shown by~\cite{CHY15}, such a PIR scheme will have rate at most $\frac{N-\frac{q}{p}}{N}$. This model was then named as the storage constrained PIR and bounds on the capacity were considered in \cite{AKT18}. Particularly, when further restricting the $\frac{p}{q}ML$ bits of storage to be uncoded, \cite{AKT18} determined the exact capacity which is achieved by a memory sharing method plus the capacity-achieving schemes of \cite{SJ17B}. Since the B-E-P scheme has asymptotically optimal rate, it is natural to consider the memory sharing method using several parallel B-E-P schemes.

Suppose that the file size is $L=t\ell$ and we divide each file $\bfx^m$ into $t$ parts of equal size $\ell$, $\{\bfx^m_1,\dots,\bfx^m_{t}\}$. We choose some $d$ parts from each file and consider them as a new subdatabase, say $\{\bfx^m_j:1\le m \le M, 1\le j \le d\}$. Then we may perform a B-E-P subscheme for this subdatabase on some $d+1$ servers.
This subscheme occupies $\frac{d}{t}ML$ bits on each of the $d+1$ servers involved and contributes $\frac{d+1}{t}L$ bits to the download cost. A combined PIR scheme by the memory sharing method, is done by just dividing the database into several subdatabases and then implementing several parallel B-E-P subschemes, each on a certain subset of servers. Note that a sufficiently large $t$ and a proper way to allocate servers for each subscheme (say, by permutations) will guarantee that each server stores exactly $\frac{p}{q}ML$ bits. Since in this scheme each server has uncoded storage, then as suggested by \cite{AKT18}, to achieve the asymptotically optimal rate, each subscheme should be implemented on either $\lceil \frac{Np}{q} \rceil$ or $\lfloor \frac{Np}{q} \rfloor$ servers.

The rate of this scheme can be calculated as follows. Altogether a proportion $\eta$ of the database is involved in subschemes on $\lceil \frac{Np}{q} \rceil$ servers and the rest proportion $1-\eta$ is involved in subschemes on $\lfloor \frac{Np}{q} \rfloor$ servers, where $\eta \lceil \frac{Np}{q} \rceil + (1-\eta) \lfloor \frac{Np}{q} \rfloor = Np/q$. The total download is then

$$ L \cdot \big( \eta \frac{\lceil \frac{Np}{q} \rceil}{\lceil \frac{Np}{q} \rceil-1} + (1-\eta) \frac{\lfloor \frac{Np}{q} \rfloor}{\lfloor \frac{Np}{q} \rfloor-1}  \big).$$

Finally, the rest $(\epsilon-\frac{p}{q})ML$ bits on each server are used for improving the access complexity via the covering code approach aforementioned.

\begin{theorem}
In a distributed storage system consisting of $N$ servers, for every rational number $\frac{1}{N}\le\frac{p}{q}\le \epsilon$, the tuple $(\Omega,\frac{Np}{q}f(\epsilon\frac{q}{p}),\epsilon)$ is achievable, where
$$\Omega= \big( \eta \frac{\lceil \frac{Np}{q} \rceil}{\lceil \frac{Np}{q} \rceil-1} + (1-\eta) \frac{\lfloor \frac{Np}{q} \rfloor}{\lfloor \frac{Np}{q} \rfloor-1}  \big)^{-1}.$$

\end{theorem}

\section{Conclusion}\label{sec:concl}

In this paper we took into consideration the access complexity of a PIR scheme. PIR schemes with low access complexity reduce the amount of data to be accessed throughout a PIR scheme and are therefore suitable for practical use.
A few methods were considered, especially ones which use covering codes.
Some of these codes were applied on known schemes. It should be noted that these methods
are not useful for all the known schemes, e.g. the one of Sun and Jafar~\cite{SJ17B}.
Finally, the problem of generalized coset weights, which will be helpful when there are multiple
queries on each server, has independent interest in coding theory.

\vspace{-1ex}
\section*{Acknowledgment}
E. Yaakobi and Y. Zhang were supported in part by the ISF grant 1817/18.
T. Etzion and Y. Zhang were supported in part by the BSF-NSF grant 2016692.
Y. Zhang was also supported in part by a Technion Fellowship.


\vspace{-1ex}


\begin{thebibliography}{20}
\bibitem{AKT18}
M.A. Attia, D. Kumar, and R. Tandon, ``The capacity of private information retrieval from uncoded storage constrained databases," arXiv:1805.04104v2, May. 2018.
\bibitem{BU16}
K. Banawan and S. Ulukus, ``The capacity of private information retrieval from coded databases,"
\emph{IEEE Trans. on Inform. Theory}, vol.\,64, no.\,3, pp.\,1945--1956, Mar. 2018.
\bibitem{BEP17}
S. Blackburn, T. Etzion, and M. Paterson, ``PIR schemes with small download complexity and low storage requirements," arXiv:1609.07027v3, Nov. 2017.
\bibitem{CHY15}
T.H. Chan, S.W. Ho, and H. Yamamoto, ``Private information retrieval for coded storage," \emph{Proc. IEEE Int. Symp. on Inf. Theory}, pp.\,2842--2846, Hong Kong, Jun. 2015.
\bibitem{CKGS98}
{B. Chor, E. Kushilevitz, O. Goldreich, and M. Sudan}, ``Private information retrieval," \emph{J.  ACM}, vol. 45, no. 6, pp. 965--981, Nov. 1998. Earlier version in FOCS 95.
\bibitem{CF85}
G. Cohen and P. Frankl, ``Good coverings of Hamming spaces with spheres", \emph{Discrete Math.}, vol.\,56, no.\,2-3, pp.\,125--131, Oct. 1985.
\bibitem{ccbook}
G. Cohen, I. Honkala, S. Litsyn, and A. Lobstein, ``Covering codes," Elsevier, 1997.
\bibitem{ckms85}
G. Cohen, M. Karpovsky, H. Mattson, and J. Schatz, ``Covering radius: Survey and recent results," \emph{IEEE Trans. on Inform. Theory}, vol.\,31, no.\,3, pp.\,328--343, May 1985.
\bibitem{GHSY12}
P. Gopalan, C. Huang, H. Simitci, and S. Yekhanin, ``On the locality of codeword symbols," \emph{IEEE Trans. on Inform. Theory}, vol.\,58, no.\,11, pp.\,6925--6934, Nov. 2012.
\bibitem{HBA98}
C.-T. Ho, J. Bruck, and R. Agrawal, ``Partial-sum queries in OLAP data cubes using covering codes," \emph{IEEE Trans. on Comp.}, vol. 47, no. 12, pp. 1326--1340, Dec. 1998.
\bibitem{Lav18}
J. Lavauzelle, ``Private information retrieval from transversal designs," \emph{IEEE Trans. on Inform. Theory}, to appear.
\bibitem{SRR14}
N.B. Shah, K.V. Rashmi, and K. Ramchandran, ``One extra bit of download ensures perfectly private information retrieval," \emph{Proc. IEEE Int. Symp. on Inform. Theory}, pp.\,856--890, Honolulu, HI, Jul. 2014.
\bibitem{SJ17B}
H. Sun and S.A. Jafar, ``The capacity of private information retrieval," \emph{IEEE Trans. on Inform. Theory}, vol.\,63, pp.\,4075--4088, July 2017.
\bibitem{SJ16C}
H. Sun and S.A. Jafar, ``Optimal download cost of private information retrieval for arbitrary message length,"
\emph{IEEE Trans. Inform. Forensics and Security}, vol.\,12, pp.\,2920--2932, Dec. 2017.
\bibitem{TWB14}
I. Tamo, Z. Wang, and J. Bruck, ``Access versus bandwidth in codes for storage," \emph{IEEE Trans. on Inform. Theory}, vol.\,60, pp.\,2028--2037, Apr. 2014.
\bibitem{TGE17}
R. Tajeddine, O.W. Gnilke, and S. El Rouayheb, ``Private information retrieval from MDS coded data in distributed storage systems," \emph{IEEE Trans. on Inform. Theory}, vol.\,64, pp.\,7081--7093, Nov. 2018.
\bibitem{Wei91}
V.K. Wei, ``Generalized Hamming weights for linear codes," \emph{IEEE Trans. on Inform. Theory}, vol.\,37, pp.\,1412--1418, Sep. 1991.
\bibitem{XZ17}
J. Xu and Z. Zhang, ``On sub-packetization of capacity-achieving PIR schemes for MDS coded databases," arXiv:1712.02466, Dec. 2017.
\end{thebibliography}
\end{document}